\documentclass[11pt]{article}

\usepackage[ruled,vlined,linesnumbered,hangingcomment]{algorithm2e}
\usepackage[colorlinks=true,linkcolor=blue,citecolor=red,urlcolor=green]{hyperref}
\usepackage[dvipsnames]{xcolor}
\usepackage{amsfonts,amssymb}
\usepackage{comment}
 
\usepackage{titlesec}
\usepackage{tikz}
\usepackage{bbding}
\usepackage{pifont}
\usepackage{titletoc}
\usepackage{hhline}
\usepackage{amsmath}
\usepackage{listings}
\usepackage{verbatim}
\usepackage{makecell}
\usepackage{tabu}
\usepackage{graphicx}
\usepackage{enumitem}
\setlist[itemize]{leftmargin=*}
\setlist[enumerate]{leftmargin=*}
\usepackage{mathrsfs}
\usepackage{setspace}
\usepackage{bbm}
\usepackage{mathtools}
\usepackage{titletoc}
\usepackage{geometry} 
\usepackage{threeparttable}
\usepackage{url} 
\usepackage{listings}
\usepackage{amssymb,amsmath,amsfonts,latexsym}
\usepackage{amsmath,graphicx,bm,xcolor,url}
\usepackage[caption=false]{subfig} 
\usepackage{array}%array and tabular environments
\usepackage{verbatim}
\usepackage{bm}
\usepackage{verbatim}
\usepackage{textcomp}
\usepackage{mathrsfs}
\usepackage{relsize}
\usepackage{subfig}
 \usepackage{amsthm}

\catcode`~=11 \def\UrlSpecials{\do\~{\kern -.15em\lower .7ex\hbox{~}\kern .04em}} \catcode`~=13 

\allowdisplaybreaks[3]

\newcommand{\calI}{\mathcal{I}}
\newcommand{\calJ}{\mathcal{J}}
\newcommand{\calK}{\mathcal{K}}

\newcommand{\calX}{\mathcal{X}}

\newcommand{\ba}{\mathbf{a}}
\newcommand{\bA}{\mathbf{A}}

\newcommand{\be}{\mathbf{e}}

\newcommand{\bI}{\mathbf{I}}

\newcommand{\bu}{\mathbf{u}}
\newcommand{\bU}{\mathbf{U}}
\newcommand{\bv}{\mathbf{v}}

\newcommand{\bx}{\mathbf{x}}
\newcommand{\bX}{\mathbf{X}}
\newcommand{\by}{\mathbf{y}}

\newcommand{\bz}{\mathbf{z}}
\newcommand{\bZ}{\mathbf{Z}}

\newcommand{\btau}{\bm{\tau}}

\newcommand{\bxi}{\bm{\xi}}

\DeclareMathOperator{\rank}{rank}

\newtheorem{theorem}{Theorem}

\newtheorem{assumption}{Assumption}
\newtheorem{corollary}{Corollary}

\newcommand{\qednew}{\nobreak \ifvmode \relax \else
      \ifdim\lastskip<1.5em \hskip-\lastskip
      \hskip1.5em plus0em minus0.5em \fi \nobreak
      \vrule height0.75em width0.5em depth0.25em\fi}

\usepackage[T1]{fontenc}
\usepackage[utf8]{inputenc}
 
\usepackage{amsthm}
\usepackage{xcolor, tikz}

\newtheorem{lem}{Lemma}
 
\newtheorem{rem}{Remark}

\usepackage{bm}
 
\DeclareMathOperator{\sign}{sign}

\newif\ifjrnotes
\jrnotestrue   % <-- set to \jrnotesfalse before final submission

\ifjrnotes
  \newcommand{\jrnote}[1]{%
    {\color{teal}[JR: #1]}%
  }
\else
  \newcommand{\jrnote}[1]{}
\fi

\title{Near-Optimal Lower Bounds on One-Bit Compressed Sensing of Approximately Sparse Signals}

\author{Junren Chen\thanks{Department of Statistics, Columbia University. (Email: \texttt{jc6315@columbia.edu})}\and Arya Mazumdar\thanks{Hal{\i}c{\i}o\u{g}lu Data Science Institute, UC San Diego. (Email: \texttt{arya@ucsd.edu})} \and Ming Yuan\thanks{Department of Statistics, Columbia University. (Email: \texttt{my2550@columbia.edu})}
}
\date{\today}

\begin{document}
\maketitle
 
\begin{abstract}
This paper provides the first near-optimal lower bounds  for one-bit compressed sensing of approximately sparse signals lying in a scaled $\ell_1$ ball, which is a commonly adopted relaxation of the exactly $k$-sparse assumption. In prior works, the best known upper bounds on uniform Euclidean error are of order $\widetilde{O}((k/m)^{1/3})$, where $m$ is the number of measurements. Under sub-Gaussian matrices, we establish nearly matching lower bounds for both the canonical one-bit compressed sensing model and the uniformly dithered model. Our argument is to first embed a small Euclidean ball into the signal set, which is straightforward for the dithered model but relies on a lifting map for the canonical model, and then construct two signals in this small ball that are separated in Euclidean distance by at least $(k/m)^{1/3}$ (up to logarithmic factor) but are indistinguishable from the binary measurements. Moreover,   our argument extends to approximately sparse signals that live in a properly scaled $\ell_q$ ball $(q\in [0,1])$, yielding a lower bound $\widetilde{\Omega}((k/m)^{\frac{2-q}{2+q}})$ that smoothly bridges the cases of exact sparsity ($q=0$) and $\ell_1$ sparsity ($q=1$). Finally, we discuss the extensions of our lower bounds to sub-Weibull matrices, adversarial bit flipping, matrix recovery, and characterize the transition to the non-sparse case. 
\end{abstract}
 
\section{Introduction}\label{sec:intro}
In one-bit compressed sensing (1bCS), one seeks to recover sparse vectors from the signs of linear measurements \cite{boufounos20081,jacques2013robust,plan2012robust,plan2013one}, that is, to recover some $\bx\in\Sigma^{n,*}_k:=\Sigma^n_k\cap \mathbb{S}^{n-1}$ from 
\begin{align}
    \by=\sign(\bA\bx)\label{1bcs}
\end{align}
and a known measurement matrix $\bA\in \mathbb{R}^{m\times n}$. Here, $\Sigma^n_k=\{\bu\in \mathbb{R}^n:\|\bu\|_0\le k\}$ denotes the cone of $k$-sparse vectors in $\mathbb{R}^n$, $\mathbb{S}^{n-1}=\{\bu\in \mathbb{R}^n:\|\bu\|_2=1\}$ denotes the  unit Euclidean sphere, and $\sign(a):=\mathbbm{1}(a\ge 0) - \mathbbm{1}(a<0)$ is applied element-wise to vectors. Note that it is standard to assume $\bx\in \mathbb{S}^{n-1}$ in 1bCS since the quantization loses all the information of signal norm. For arbitrary matrix $\bA$, any estimator $\hat{\bx}$ that is a measurable function of $(\bA,\by)$ admits a uniform recovery $\ell_2$ error lower bound
\begin{align}\label{1bcsexistlower}
    \sup_{\bx\in \Sigma^{n,*}_k}\|\hat{\bx}-\bx\|_2 \gtrsim \frac{k}{m}. 
\end{align} 
See \cite[Theorem 2]{jacques2013robust} and \cite[Theorem 6]{acharya2017improved}. 
Throughout this paper, we write $T_1 =\Omega(T_2)$ or $T_1\gtrsim T_2$ to denote $T_1\ge cT_2$ for some universal constant $c>0$.

Under Gaussian design $\bA$ that has i.i.d. $N(0,1)$ entries, it was shown \cite{matsumoto2024binary} that an efficient algorithm called normalized binary iterative hard thresholding (NBIHT) achieves the lower bound (\ref{1bcsexistlower}) up to logarithmic factors, i.e.,
\begin{align}\label{nbihtrate}
    \sup_{\bx\in \Sigma^{n,*}_k}\|\hat{\bx}_{{\rm nbiht}} - \bx\|_2 =\widetilde{O}\Big(\frac{k}{m}\Big),
\end{align}
where $\hat{\bx}_{{\rm nbiht}}$ denotes the output of NBIHT. (See also the earlier analysis in \cite{friedlander2021nbiht}.) By convention, $T_1=O(T_2)$ or $T_1\lesssim T_2$ denotes that $T_1\le CT_2$ for some universal constant $C>0$, and we use $\widetilde{O}(\cdot)$ and $\widetilde{\Omega}(\cdot)$ to hide logarithmic factors of $(m,n)$ in $O(\cdot)$ and $\Omega(\cdot),$ respectively. Throughout the paper, we use bold letters for vectors and matrices, regular letters for scalars.

However, signals in real applications are typically only approximately sparse rather than exactly living in $\Sigma^n_k$ for some $k$. For any $q>0$, let $\mathbb{B}_q^n=\{\bu\in \mathbb{R}^n:\|\bu\|_q:= (\sum_{i=1}^n|u_i|^q)^{1/q}\le 1\}$ denote the unit $\ell_q$ ball. The most common formulation of the set of approximately (or effectively) $k$-sparse signals in the unit Euclidean ball $\mathbb{B}_2^n$ is via
\begin{align*}
    \calK_{1,k} := \sqrt{k} \mathbb{B}_1^n \cap \mathbb{B}_2^n. 
\end{align*}
This is a relaxation of $\Sigma^n_k\cap \mathbb{B}_2^n$ in view of $\calK_{1,k}\supset\Sigma^n_k\cap \mathbb{B}_2^n$, and indeed, $\calK_{1,k}$ can be viewed as the convex hull of $\Sigma^n_k\cap \mathbb{B}_2^n$ \cite[Lemma 3.1]{plan2013one}. In the sequel, we may refer to the vectors in $\calK_{1,k}$ as $\ell_1$-sparse signals, implicitly for some sparsity level $k$.

Motivated by this consideration, previous works have established a number of recovery guarantees for 1bCS of $\bx\in \calK_{1,k}^*:= \sqrt{k}\mathbb{B}_1^n\cap \mathbb{S}^{n-1}$ \cite{plan2012robust,plan2013one,awasthi2016learning,chinot2022adaboost,chen2024optimal,dirksen2020one}. For simplicity we assume $k$ is known. The sharpest error upper bound  exhibits a non-standard decay rate of $O(m^{-1/3})$, precisely,
\begin{align} \label{l1sparserate}
    \sup_{\bx\in \calK_{1,k}^*} \|\hat{\bx} -\bx\|_2 =\widetilde{O}\Big(\big(\frac{k}{m}\big)^{1/3}\Big),
\end{align}
where $\hat{\bx}$ is some estimator constructed from $(\bA,\by)$. Under Gaussian matrix $\bA$, this error rate is achieved by constrained hamming distance minimization \cite{oymak2015near,chen2024optimal}
\begin{align} \label{hdm1bcs}
    \hat{\bx}_{{\rm hdm}} := \textrm{arg}\min_{\bu\in\calK_{1,k}^*}\, d_{\rm H}(\sign(\bA\bu),\by)
\end{align}
that is arguably the best possible decoder but is computationally intractable; here, we define the hamming distance between two binary vectors $\bu,\bv\in \{-1,1\}^m$ as $d_{\rm H}(\bu,\bv):=\sum_{i=1}^m \mathbbm{1}(u_i\ne v_i)$. Fortunately, there also exist two efficient algorithms, Adaboost \cite{chinot2022adaboost} and an $\ell_1$ projected gradient descent algorithm \cite{chen2024optimal},   achieving the error rate $\widetilde{O}((k/m)^{1/3})$. The main focus of this paper is on lower bound, hence we simply refer to \cite[Remark 6]{chen2024optimal} for the attainability of (\ref{l1sparserate}) via (\ref{hdm1bcs}) under $m=\widetilde{\Omega}(k)$.

As such, while relaxing the exact sparsity $\bx\in \Sigma^{n,*}_k$ to the $\ell_1$ sparsity $\calK_{1,k}^*$, the sharpest known upper bounds exhibit essential degradation, particularly from the provably optimal $\widetilde{O}(k/m)$ in (\ref{nbihtrate}) to $\widetilde{O}((k/m)^{1/3})$ in (\ref{l1sparserate}). To our best  knowledge, it remains an open question whether the upper bound for 1bCS of $\ell_1$ sparse vectors in (\ref{l1sparserate}) is information-theoretically optimal. Indeed, it was written \cite[Page 131]{chinot2022adaboost} that ``it is not clear whether the exponent in $O((k/m)^{1/3})$ is optimal, and further research about information theoretic lower bounds is needed.'' See also \cite[Remark 2]{chen2024optimal}. Before this work,  $\Omega(k/m)$ in  (\ref{nbihtrate}) appears to  be the only known lower bound for 1bCS. While this lower bound trivially applies to 1bCS of $\bx\in \calK_{1,k}^*$, it does not match the best known upper bound $\widetilde{O}((k/m)^{1/3})$.

As we emphasize, in the canonical model of 1bCS (\ref{1bcs}), all signal norm information is absorbed into the $\sign$ quantizer, hence there is no hope to recover the signal norm. To the end of signal norm recovery, it was proposed to use \emph{dithering}, a technique of adding ``artificial noise'' $\btau \in \mathbb{R}^m$---referred to as ``dither''---prior to the quantization:
\begin{align}
    \label{d1bcs}
    \by = \sign(\bA\bx + \btau). 
\end{align}
We refer to such a model with non-zero $\btau$ as dithered one-bit compressed sensing (D1bCS).
Unlike traditional noise, one can design  the generation of $\btau$ and has full access to the its entries for signal reconstruction. Under some dithering level $\lambda>0$, existing works studied (for instance) Gaussian dither $\btau\sim N(0,\lambda^2\bI_m)$ \cite{knudson2016one} and uniform dither $\btau\sim {\rm Unif}[-\lambda,\lambda]^m$ \cite{dirksen2021non,thrampoulidis2020generalized,xu2020quantized}.

Similar open question exists for D1bCS of $\ell_1$ sparse signals, that is, the recovery of $\bx\in \calK_{1,k}=\sqrt{k}\mathbb{B}_1^n\cap \mathbb{B}_2^n$ from $(\bA,\btau,\by)$ where $\by=\sign(\bA\bx+\btau)$.\footnote{In D1bCS, we no longer need to assume $\bx\in \mathbb{S}^{n-1}$. We assume $\bx\in\mathbb{B}_2^n$ but note that this can be easily generalized to $\bx\in \lambda_0\mathbb{B}_2^n$ for some $\lambda_0>0$.} For concreteness, we consider the setting with sub-Gaussian $\bA$ (i.e., the rows of $\bA$ are i.i.d. isotropic sub-Gaussian vectors) and uniform dither $\btau \sim {\rm Unif}[-\lambda,\lambda]^m$.\footnote{Although we often keep track of $\lambda$ in the error rates and sample complexities for D1bCS, the regime of main interest is when $\lambda\ge C$ for some large enough universal constant $C$ and $\lambda\le {\rm Polylog}\,n$, since most existing upper bounds hold in this regime and worsen if $\lambda$ is overly large \cite{dirksen2021non,thrampoulidis2020generalized,dirksen2023robust,chen2024optimal,jung2021quantized}.} Denote by $\hat{\bx}$ some decoder constructed from $(\bA,\btau,\by)$, then the best known upper bound is 
\begin{align}
    \label{d1bcsupper}
    \sup_{\bx\in \calK_{1,k}}\| \hat{\bx}-\bx\|_2 =\widetilde{O}\Big(\big(\frac{\lambda k}{m}\big)^{1/3}\Big).
\end{align}
This upper bound is attained by the computationally infeasible decoder of hamming distance minimization \cite{dirksen2021non,chen2024optimal} 
\begin{align}\label{hdmd1bcs}
    \hat{\bx}_{{\rm hdm}} := \textrm{arg}\min_{\bu\in \calK_{1,k}}\,d_{\rm H}\big(\sign(\bA\bu+\btau),\by\big),
\end{align}
as well as by two efficient algorithms of convex program \cite{jung2021quantized} and projected gradient descent \cite{chen2024optimal}. Particularly, (\ref{d1bcsupper}) is achieved by (\ref{hdmd1bcs}) under $m=\widetilde{\Omega}(\lambda^{-2}k)$; see, e.g., \cite[Remark 6]{chen2024optimal}.
To our best knowledge, it remains open whether the upper bound (\ref{d1bcsupper}) is tight.

The prior works most relevant to the open questions (of the tightness of $(k/m)^{1/3}$) are \cite{dirksen2022sharp} and \cite{dirksen2026resolution}, which proved a lower bound $\widetilde{\Omega}((k/m)^{1/3})$ for \emph{uniform hyperplane tessellation}. In particular, the result of Dirksen, Mendelson and Stollenwerk \cite[Theorem 1.10]{dirksen2022sharp} has the following implication: under Gaussian $\bA$ and uniform dither $\btau\sim {\rm Unif}[-\lambda,\lambda]^m$ with large enough $\lambda$, with high probability it holds that 
\begin{align}\label{d1bcsteselower}
    \sup_{\bu,\bv\in \calK_{1,k}}\bigg|\frac{\sqrt{2\pi}\lambda}{m}d_{\rm H}\big(\sign(\bA\bu+\btau),\sign(\bA\bv+\btau)\big)-\|\bu-\bv\|_2\bigg| \gtrsim \Big(\frac{\lambda k}{m}\Big)^{1/3}.
\end{align}
Yet this does not imply the sharpness of the uniform recovery upper bound in (\ref{d1bcsupper}) for D1bCS. In fact, the argument of \cite{dirksen2022sharp} is to construct, via Dvoretzky--Milman theorem (e.g., \cite{ArtsteinAvidanGiannopoulosMilman2015}), $\bu^*\in \calK_{1,k}$ such that 
\begin{align} \label{dirksenconstruct}
    \|\bu^*\|_2\lesssim \Big(\frac{\lambda k}{m}\Big)^{1/3},\quad  \frac{\lambda}{m}d_{\rm H}\big(\sign(\bA\bu^*+\btau),\sign(\btau)\big)\gtrsim \Big(\frac{\lambda k}{m}\Big)^{1/3}
\end{align}
up to logarithmic factors, and then set $(\bu,\bv)=(\bu^*,\bm{0})$ on the left-hand side of (\ref{d1bcsteselower}) to yield the lower bound. However, the large hamming distance in such construction indeed helps the recovery of $\bu^*$. More recently, Dirksen and Strachan \cite{dirksen2026resolution} adapted this argument to 1bCS via a ``lifting'' trick (see also \cite{plan2014dimension}) and showed that
\begin{align}\label{1bcstesselower}
    \sup_{\bu,\bv\in \calK_{1,k}^*}\bigg|\frac{1}{m}d_{\rm H}\big(\sign(\bA\bu),\sign(\bA\bv)\big) -d_{\mathbb{S}^{n-1}}(\bu,\bv)\bigg| \gtrsim \Big(\frac{k}{m}\Big)^{1/3} 
\end{align}
holds with high probability on the Gaussian $\bA$, where $d_{\mathbb{S}^{n-1}}(\bu,\bv):=\pi^{-1}\arccos(\langle \bu,\bv\rangle)$ denotes the geodesic distance over $\mathbb{S}^{n-1}$. But this sheds no light on the tightness of (\ref{l1sparserate}); indeed, like \cite{dirksen2022sharp}, the construction of the pair of $(\bu,\bv)$ achieving the lower bound in (\ref{1bcstesselower}) is in a ``large-hamming-distance'' direction that benefits the recovery.

This paper resolves the two open questions and establishes near-optimal lower bounds for (D)1bCS of approximately sparse signals. Our new lower bounds match  the upper bounds in (\ref{l1sparserate}) and (\ref{d1bcsupper}), up to logarithmic factors, and show that (with high probability) no decoder can achieve a substantially smaller uniform recovery error. Like \cite{dirksen2022sharp,dirksen2026resolution}, we first resolve D1bCS and then transfer the developments to 1bCS without dithering (\ref{1bcs}) via a lifting trick (see Lemma \ref{lem:lifting}). However, while our arguments are also constructive in nature, the construction is toward an entirely  converse direction that renders recovery impossible. Specifically, for the dithered model $\by=\sign(\bA\bx+\btau)$, we construct $\bu^*\in\calK_{1,k}$ such that 
\begin{align}
    \|\bu^*\|_2\gtrsim \Big(\frac{\lambda k}{m}\Big)^{1/3},\quad \sign(\bA\bu^*+\btau) = \sign(\btau)\label{ourconstruct}
\end{align}
up to logarithmic factors. This is in stark contrast to (\ref{dirksenconstruct}) and ensures that $\bu^*$ and $\bm{0}$---which are  $\widetilde{\Omega}((\lambda k/m)^{1/3})$-separated in Euclidean distance---are indistinguishable, yielding the matching lower bound for (\ref{d1bcsupper}).

At first glance, it appears difficult to construct $\bu^*$ that satisfies (\ref{ourconstruct}), especially in view of the exact consistency \[\sign(\bA\bu^*+\btau)=\sign(\btau)\] that requires considerations on all $m$ measurements. (In contrast, to render $\frac{\lambda}{m}d_{\rm H}(\sign(\bA\bu^*+\btau),\sign(\btau))\gtrsim (\frac{\lambda k}{m})^{1/3}$, the authors of \cite{dirksen2022sharp} could simply focus on a subset of measurements with small $\tau_i$.) However, our construction is even conceptually simpler than \cite{dirksen2022sharp} (see Lemma \ref{lem:construct}) and readily applies to more general contexts such as recovery of $\ell_q$-sparse signals, sub-Weibull designs and low-rank matrix sensing (see Section \ref{sec:extension}). Moreover, our results reproduce (\ref{d1bcsteselower}) from \cite{dirksen2022sharp} and (\ref{1bcstesselower}) from \cite{dirksen2026resolution}, up to logarithmic factors; see Remarks \ref{rem:reproduce1}, \ref{rem:reproduce2}.

The rest of this paper is organized as follows. In Section \ref{sec:mainresults} we present our main results and put them in perspective. Section \ref{sec:proof} provides the complete proofs to our results. We then discuss multiple extensions of our results in Section \ref{sec:extension} and conclude the paper in Section \ref{sec:conclusion}.

%We note in passing that the results of \cite{dirksen2022sharp,dirksen2026resolution} apply to general sets.%, we shall specialize to approximately sparse signals here. 

\section{Main Results}\label{sec:mainresults}
Our main results are the matching lower bounds for (\ref{d1bcsupper}) and (\ref{l1sparserate}).

\subsection{Dithered One-Bit Compressed Sensing}
We start with the dithered model $\by=\sign(\bA\bx+\btau)$. We first assume the sub-Gaussianity of the sensing vectors. %without requiring that they are independent.
\begin{assumption} \label{assump1}
For some $L>0$, the rows of $\bA$---denoted by $\ba_1,\cdots,\ba_m$---are independent and $L$-sub-Gaussian: for any $i\in[m]:=\{1,2,\cdots,m\}$, $\bu\in \mathbb{S}^{n-1}$, and $t>0$, 
\begin{align} 
\label{subgauss1}
\mathbb{P}\big(|\ba_i^\top  \bu|\ge t\big) \le 2\exp\Big(-\frac{t^2}{2L^2}\Big). 
\end{align}
\end{assumption}

We then make the following assumption on the dither. For instance, (\ref{smallball}) is satisfied if $\tau_i$'s are continuous random variables with density function uniformly bounded by $\rho_\tau$. 

\begin{assumption}\label{assump2}
$\btau$ is independent of $\bA$,
  and for some $\rho_{\tau}>0$, the entries of $\btau$---denoted by $\tau_1,\cdots,\tau_m$---are independent and satisfy a small-ball probability condition: for any $i\in [m]$ and $t>0$, 
    \begin{align}\label{smallball}
        \mathbb{P}\big(|\tau_i|\le t\big)\le 2\rho_{\tau}t.
    \end{align}
\end{assumption}

Throughout this paper, we use $C_1,C_2,c_1,c_2,\cdots$ to denote universal constants whose values can vary from line to line. We   use $\lfloor a \rfloor$ to denote the largest integer not greater than $a$. We also define $a\wedge b:=\min\{a,b\}$ and $a\vee b: = \max\{a,b\}$.

Our first main result concerns the lower bound on the uniform recovery of $\bx\in \calK_{1,k}$ from $\by=\sign(\bA\bx+\btau)$. 
\begin{theorem}[Lower bound for D1bCS over $\calK_{1,k}$] \label{thm:d1bcslower}
    Assume that Assumptions \ref{assump1}--\ref{assump2} hold, and assume $k\ge 2$. For any $\sqrt{\frac{k}{m\wedge n}}\le r\le 1$, if
    \begin{align}\label{thm1con}
        \frac{k}{m\sqrt{\log m}}\ge 32\rho_{\tau}r^3L,  
    \end{align}
then    
    with probability at least $1-\frac{2}{m}-\exp(-\frac{k}{24})$, any decoder $\hat{\bx}$ that is a measurable function of $(\bA,\btau,\by=\sign(\bA\bx+\btau))$ satisfies 
    \begin{align}\label{thm1d1bcsconclu}
        \sup_{\bx\in \calK_{1,k}}\|\hat{\bx}-\bx\|_2 \ge \frac{r}{2}. 
    \end{align}
\end{theorem}

We now specialize Theorem \ref{thm:d1bcslower} to the   setting with uniform dither $\tau \sim {\rm Unif}[-\lambda,\lambda]^m$, where the upper bound in (\ref{d1bcsupper}) is achieved over rather general  sub-Gaussian matrices $\bA$ \cite{dirksen2021non,jung2021quantized,chen2024optimal}. %With some algebra this immediately yields the following statement, which indicates the sharpness of (\ref{d1bcsupper}). 

\begin{corollary} \label{cor:d1bcssharp}
    Assume that Assumption \ref{assump1} holds,      $\btau\sim {\rm Unif}[-\lambda,\lambda]^m$  is independent of $\bA$, $k\ge 2$, 
    \begin{align}\label{cond1bcs}
        \lambda\ge \frac{16Lm\sqrt{k\log m}}{(m\wedge n)^{3/2}}\qquad \textrm{and}\qquad 16Lm\sqrt{\log m}\ge \lambda k.
    \end{align}
    Then with probability at least $1-\frac{2}{m}-\exp(-\frac{k}{24})$, any decoder $\hat{\bx}$ that is a measurable function of $(\bA,\btau,\by=\sign(\bA\bx+\btau))$ satisfies 
    \begin{align*}
        \sup_{\bx\in \calK_{1,k}}\|\hat{\bx}-\bx\|_2 \ge \frac{1}{4}\bigg(\frac{\lambda k}{2Lm\sqrt{\log m}}\bigg)^{1/3}. 
    \end{align*}
\end{corollary}
\begin{rem}\label{rem:d1bcsbenigh}
   The above statement indicates that the upper bound in (\ref{d1bcsupper}) is sharp, up to logarithmic factors (and the dependence on $L$ that is not made explicit in (\ref{d1bcsupper})). We point out that the conditions in (\ref{cond1bcs}) and not restricted. In fact, ignoring $L$ and the logarithmic factor $\sqrt{\log m}$, the two conditions become 
   \begin{align}
       \lambda \gtrsim \frac{m\sqrt{k}}{(m\wedge n)^{3/2}}\qquad \textrm{and}\qquad m\gtrsim \lambda k.\label{d1bcsconsv2} 
       \end{align}
   In sparse recovery,  the regime of main interests is  $k\ll n$ and $k\lesssim m\le n$, under which $\frac{m\sqrt{k}}{(m\wedge n)^{3/2}}$ typically reads $\sqrt{k/m}$ and is hence  small enough. In this regime, (\ref{d1bcsconsv2}) covers $\lambda = \widetilde{O}(1)$ and $m=\widetilde{\Omega}(k)$, which is also the regime in which the matching upper bound (\ref{d1bcsupper}) holds. %under $\lambda=\widetilde{O}(1)$ that is of main interest.
   %These are exactly the conditions for the achievability of (\ref{d1bcsupper}) (see Appendix \ref{app:background}).  
\end{rem}

\begin{rem}\label{rem:reproduce1}
    As will be clear from the next section, the main accomplishment of our proof is to construct $\bu^*,\bv^*\in \calK_{1,k}$ satisfying $\|\bu^*-\bv^*\|_2 \gtrsim (\frac{\lambda k}{Lm\sqrt{\log m}})^{1/3}$ and $\sign(\bA\bu^*+\btau)=\sign(\bA\bv^*+\btau)$. As such, we readily recover the hyperplane tessellation lower bound in (\ref{d1bcsteselower}) up to logarithmic factors, in the light of
    \begin{align*}
        &\sup_{\bu,\bv\in \calK_{1,k}}\bigg|\frac{\sqrt{2\pi}\lambda}{m}d_{\rm H}\big(\sign(\bA\bu+\btau),\sign(\bA\bv+\btau)\big)-\|\bu-\bv\|_2\bigg| \\
        &\,\gtrsim \bigg|\frac{\sqrt{2\pi}\lambda}{m}d_{\rm H}\big(\sign(\bA\bu^*+\btau),\sign(\bA\bv^*+\btau)\big)-\|\bu^*-\bv^*\|_2\bigg|\\
        &\,=\|\bu^*-\bv^*\|_2 \gtrsim \bigg(\frac{\lambda k}{Lm\sqrt{\log m}}\bigg)^{1/3}. 
    \end{align*} 
\end{rem}

\subsection{One-Bit Compressed Sensing} 
We now move on to 1bCS without dithering (\ref{1bcs}), where one seeks to recover $\bx\in \calK_{1,k}^*=\sqrt{k}\mathbb{B}_1^n\cap \mathbb{S}^{n-1}$ from $\by=\sign(\bA\bx)$. We make the following assumption on $\bA$. 

\begin{assumption} \label{assump3}
    The rows of $\bA$---denoted by $\ba_1,\cdots,\ba_m$---are independent. For $i\in[m]$, write $\ba_i=(\ba_{i,\setminus n}^\top ,a_{in})^\top $ where $\ba_{i,\setminus n}\in \mathbb{R}^{n-1}$ is the vector of the first $n-1$ entries of $\ba_i$ and $a_{in}$ denotes the last entry. For some $L,\rho>0$, assume for any $i\in [m]$, $t>0$ that 
    \begin{gather*}
        \sup_{\bu\in \mathbb{S}^{n-2}}\mathbb{P}\bigg(|\ba_{i,\setminus n}^\top \bu|\ge t \biggm|a_{in}\bigg)\le 2\exp\Big(-\frac{t^2}{2L^2}\Big),
        \\
        \mathbb{P}\big(|a_{in}|\le t\big)\le 2\rho t.
    \end{gather*}
\end{assumption}

Our main result for canonical 1bCS without dithering (\ref{1bcs}) is as follows. 

\begin{theorem}[Lower bound for 1bCS over $\calK_{1,k}^*$] \label{thm:1bcslower} 
Assume that Assumption \ref{assump3} holds, and assume $k\ge 8$. For any $\frac{\sqrt{k}}{\sqrt{(m\wedge n)-1}}\le r\le 1$, if 
\begin{align}\label{1bcsscaconcon}
    \frac{k}{m\sqrt{\log m}}\ge 64\rho r^3L, 
\end{align}
then with probability at least $1-\frac{2}{m}-\exp(-\frac{k}{48})$, any decoder $\hat{\bx}$ that is a measurable function of $(\bA,\by=\sign(\bA\bx))$ satisfies
\begin{align} \label{uniformlowerthm2}
    \sup_{\bx\in \calK_{1,k}^*}\|\hat{\bx}-\bx\|_2\ge \frac{r}{2\sqrt{2}}.
\end{align}
\end{theorem}

We specialize the above result to the recovery of $\bx\in \calK^*_{1,k}$ from $\by=\sign(\bA\bx)$, where $\bA\sim N^{m\times n}(0,1)$. In this case, Assumption \ref{assump3} holds with  $L=1$ and $\rho = 1/\sqrt{2\pi}$. 

\begin{corollary}\label{cor:1bcssharp}
    Consider the recovery of $\bx \in \calK^*_{1,k}$ from $\by=\sign(\bA\bx)$, where $\bA\sim N^{m\times n}(0,1)$ and $k\ge 8$. If
    \begin{align} \label{1bcscons}
        \frac{m\sqrt{k\log m}}{((m\wedge n)-1)^{3/2}}\le\frac{\sqrt{2\pi}}{64},\quad\textrm{and}\quad m\sqrt{\log m}\ge \frac{\sqrt{2\pi}k}{64} 
    \end{align}
    %\am{arya: the inequality should be reversed?}
    then with probability at least $1-\frac{2}{m}-\exp(-\frac{k}{48})$, any decoder $\hat{\bx}$ that is a measurable function of $(\bA,\by=\sign(\bA\bx))$ satisfies 
    \begin{align*}
        \sup_{\bx\in \calK^*_{1,k}}\|\hat{\bx}-\bx\|_2\ge \frac{1}{8}\Big(\frac{\pi}{4}\Big)^{1/6}\bigg(\frac{k}{m\sqrt{\log m}}\bigg)^{1/3}. 
    \end{align*}
\end{corollary}
\begin{rem}
    Corollary \ref{cor:1bcssharp} indicates that the uniform error rate in (\ref{l1sparserate}) is sharp, up to logarithmic factors. The conditions in (\ref{1bcscons}) are benign---the first condition is typically satisfied in sparse recovery (as explained in Remark \ref{rem:d1bcsbenigh}), while the second condition, up to logarithmic factors, is exactly the sample complexity for the matching upper bound (\ref{l1sparserate}) to hold with high probability. 
\end{rem}
\begin{rem}\label{rem:reproduce2}
    The core of the argument is to construct $\bu^*,\bv^*\in \calK_{1,k}^*$ that are distant but indistinguishable. Specifically, in the setting of Corollary \ref{cor:1bcssharp}, it is shown that $\bu^*,\bv^*\in \calK_{1,k}^*$ obeying $\|\bu^*-\bv^*\|_2\gtrsim (\frac{k}{m\sqrt{\log m}})^{1/3}$ and $\sign(\bA\bu^*)=\sign(\bA\bv^*)$ exist with high probability. Similarly to Remark \ref{rem:reproduce1},  this implies 
    \begin{align*}
         \sup_{\bu,\bv\in \calK_{1,k}^*}\bigg|\frac{1}{m}d_{\rm H}\big(\sign(\bA\bu),\sign(\bA\bv)\big) -d_{\mathbb{S}^{n-1}}(\bu,\bv)\bigg|  \gtrsim \bigg(\frac{k}{m\sqrt{\log m}}\bigg)^{1/3}.
    \end{align*}
    %\am{arya: normalization by $m$?}
    This reproduces (\ref{1bcstesselower}) from \cite{dirksen2026resolution}, up to logarithmic factors. 
\end{rem}

\section{Proofs}\label{sec:proof}
In this section, we provide the complete proofs for the results in Section \ref{sec:mainresults}. For a matrix $\bA\in \mathbb{R}^{m\times n}$ (which reduces to a vector when $n=1$) and the index sets $\calI \subset [m]$ and $\calJ\subset [n]$, we use $\bA^{\calI}_{\calJ}\in \mathbb{R}^{|\calI|\times |\calJ|}$ to denote the sub-matrix of $\bA$ that retains the rows in $\calI$ and columns in $\calJ$ only. If retaining either all the $m$ rows or all the $n$ columns, we write $\bA^{\calI}:= \bA^{\calI}_{[n]}$ and $\bA_{\calJ}:=\bA^{[m]}_{\calJ}$.

The main technical ingredient is the following ``construction lemma'', which   constructs a point in a radius-$r$  Euclidean sphere in $\mathbb{R}^s$ that is indistinguishable from $0$ in the dithered model. As we shall see, this immediately yields Theorem \ref{thm:d1bcslower} by embedding a suitable Euclidean ball into $\calK_{1,k}$, and when combined with a lifting trick, it also yields Theorem \ref{thm:1bcslower}. Corollaries \ref{cor:d1bcssharp}, \ref{cor:1bcssharp} then directly follow from Theorems \ref{thm:d1bcslower}, \ref{thm:1bcslower}.

In particular, by restricting to an $s$-dimensional space, we denote by $\bZ\in \mathbb{R}^{m\times s}$ the sensing matrix, and by $\bxi \in \mathbb{R}^m$ the ``dither''. Our construction of $\bu^*\in \mathbb{S}^{s-1}$ such that 
\begin{align}
    \sign(\bZ(r\bu^*)+\bxi)=\sign(\bxi)\label{exactconsistent}
\end{align} is conceptually simple: 
\begin{enumerate}
    \item Observe that, to ensure  (\ref{exactconsistent}), it is sufficient to  have $|\bZ\bu^*|<|\frac{\bxi}{r}|$, where $|\cdot|$ and $<$ apply elementwise, we  separately treat the entries of $|\frac{\bxi}{r}|$ that are ``smaller'' or ``larger'' than $2L\sqrt{\log m}$;

    \item   We show that under the condition in (\ref{goodcondition}), with high probability, the number of ``small'' entries is lower than $s-1$, hence we can construct $\bu^*\in \mathbb{S}^{s-1}$ which is orthogonal to the corresponding rows in $\bZ$ to ensure (\ref{exactconsistent}) for these smaller entries of $|\frac{\bxi}{r}|$; 

    \item  The remaining ``large'' entries of $|\frac{\bxi}{r}|$ are greater than $2L\sqrt{\log m}$ and  are therefore greater than the corresponding entries of $|\bZ\bu^*|$, in light of a standard sub-Gaussian bound.
\end{enumerate}
 For similar idea in related problems, see \cite[Theorem 3.5]{chen2024robust} for instance.

\begin{lem}[Construction lemma]\label{lem:construct}
    Given $r\in (0,1]$ and integer $1\le s\le m$, assume that $\bxi\in \mathbb{R}^m$ has independent entries $\{\xi_i:i\in[m]\}$ satisfying 
    \begin{align}\label{smallballxi}
        \mathbb{P}\big(|\xi_i| \le t\big)\le 2\rho t,\quad \forall t>0
    \end{align}
    for some $\rho>0$. Also assume that, when conditioning on $\bxi$, the rows of $\bZ\in \mathbb{R}^{m\times s}$---denoted by $\bz_1,\cdots,\bz_m$---are independent and satisfy 
    \begin{align} \label{consg}
        \sup_{\bu\in \mathbb{S}^{s-1}}\mathbb{P}\bigg(|\bz_i^\top \bu|\ge t\biggm|\bxi\bigg)\le 2\exp\Big(-\frac{t^2}{2L^2}\Big),\quad \forall t>0
    \end{align}
    for some $L>0$.  If 
    \begin{align}\label{goodcondition}
        \frac{s}{m\sqrt{\log m}}\ge 16\rho rL, 
    \end{align}
    then with probability at least $1-\frac{2}{m}-\exp(-\frac{s}{12})$, there exists $\bu^*\in \mathbb{S}^{s-1}$ such that 
    \begin{align}
        \sign(\bZ(r\bu^*)+\bxi)=\sign(\bxi). \label{desired}
    \end{align}
\end{lem}
\begin{proof}
     A sufficient condition for (\ref{desired}) is $|\bZ(r\bu^*)|<|\bxi|$, or equivalently, 
     \begin{align}\label{sufficecon}
         |\bZ\bu^*|< \Big|\frac{\bxi}{r}\Big|,
     \end{align}
     where $|\cdot|$ and $<$ apply elementwise. For any $i\in[m]$, (\ref{smallballxi}) and (\ref{goodcondition}) give 
     \begin{align}\label{barPbound}
         \bar{P}_i:=\mathbb{P}\bigg(\Big|\frac{\xi_i}{r}\Big|\le 2L\sqrt{\log m}\bigg)\le 4\rho rL\sqrt{\log m} \le \frac{s}{4m}\le \frac{1}{4}.
     \end{align}
     Since the entries of $\bxi$ are independent, the cardinality of 
     \begin{align}\label{defI}
         \calI:=\{i\in[m]:|\xi_i/r|\le 2L\sqrt{\log m}\} 
     \end{align} is stochastically dominated by a binomial variable $ {\rm Binomial}(m,\frac{s}{4m}),$ hence Chernoff bound yields
     \begin{align*}
         \mathbb{P}\Big(|\calI|\le\frac{s}{2}\Big) \ge \mathbb{P}\bigg({\rm Binomial}\Big(m,\frac{s}{4m}\Big)\le\frac{s}{2}\bigg) \ge 1- \exp\Big(-\frac{s}{12}\Big).
     \end{align*}
     From now on, we condition on $\bxi$ and assume that \begin{align}
         |\calI|\le \frac{s}{2}, \label{Ismallset}
     \end{align} which holds with the promised probability. To guarantee (\ref{sufficecon}), it is sufficient to ensure
     \begin{gather}\label{suffI}
         |\bZ^{\calI}\bu^*|<\Big|\frac{\bxi^{\calI}}{r}\Big|,\\\label{suffIc}
         |\bZ^{\calI^c}\bu^*|<\Big|\frac{\bxi^{\calI^c}}{r}\Big|,
     \end{gather}
     where $\calI^c:=[m]\setminus \calI$. By (\ref{Ismallset}), $\bZ^{\calI}$ is a ``fat'' matrix with more columns than rows, hence we can choose 
     \begin{align}
         \bu^*\in {\rm Ker}(\bZ^\calI)\cap \mathbb{S}^{s-1}.\label{constructustar}
     \end{align}
     This choice of $\bu^*$ ensures $|\bZ^{\calI}\bu^*| = 0$, and since (\ref{smallballxi}) gives $\mathbb{P}(\xi_i=0)=0$, (\ref{suffI}) now holds almost surely. It remains to show (\ref{suffIc}). By the definition of $\calI$ in (\ref{defI}), 
     \begin{align*}
         \Big|\frac{\bxi^{\calI^c}}{r}\Big| > 2L\sqrt{\log m},
     \end{align*}
     and hence it suffices to have
     \begin{align}\label{suffdesire}
         |\bZ^{\calI^c}\bu^*| \le 2L\sqrt{\log m}.
     \end{align}
     By assumption, when conditioning on $\bxi$, the rows of $\bZ$ are independent. Hence, $\bu^*$ constructed from $\bZ^{\calI}$ is independent of $\bZ^{\calI^c}$, and therefore (\ref{consg}) gives 
     \begin{align*}
         \mathbb{P}\bigg(|\bz_i^\top  \bu^*|\ge t\biggm| \bxi\bigg)\le 2\exp\Big(-\frac{t^2}{2L^2}\Big),\quad \forall t>0 
     \end{align*}
     holds for any $i\in \calI^c$. Therefore, setting $t=2L\sqrt{\log m}$ and taking a union bound over $i\in \calI^c$ establish that, conditioning on $\bxi$,
     \begin{align} \label{sgmaxupper}
         \max_{i\in \calI^c}|\bz_i^\top \bu^*| \le 2L\sqrt{\log m}
     \end{align}
     holds with probability at least $1-\frac{2}{m}$. This is exactly the desired (\ref{suffdesire}). We can now conclude the proof, as we have found $\bu^*$ such that (\ref{suffI})--(\ref{suffIc}) hold with the promised probability.  
\end{proof}

\subsection{Technical Overview}\label{sec:technicaloverview}
Before proceeding to the detailed proofs, we pause to provide a technical overview, with an emphasis on the two key components---an \emph{embedding condition} (cf. (\ref{l1dcon1}) and (\ref{1bcscon1a}) below) and a \emph{construction condition} (cf. (\ref{l1dcon2}) and (\ref{1bcscon1b}) below)---that dictate the non-standard $(k/m)^{1/3}$ rate:
 
\begin{enumerate}
    \item (Theorem \ref{thm:d1bcslower} \& Corollary \ref{cor:d1bcssharp}) In D1bCS with signal space $\calK_{1,k}$, by identifying $\mathbb{R}^s$ with the first $s$ coordinates of $\mathbb{R}^n$,  we first embed the radius-$r$ $s$-dimensional Euclidean ball $r\mathbb{B}_2^s$ into the signal space
    \begin{align*}
        r\mathbb{B}_2^s\subset \calK_{1,k}
    \end{align*}
    so long as 
    \begin{align}\label{l1dcon1}
        s \le  \frac{k}{r^2}.
    \end{align}
    We then restrict our attention to $r\mathbb{B}_2^s$ and seek to find a point in $r\mathbb{S}^{s-1}$ that is indistinguishable from $0$.\footnote{An intuition why the restriction to $r\mathbb{B}_2^s$ does not lose too much is that, under $s\asymp \frac{k}{r^2}$, the Gaussian width of $r\mathbb{B}_2^s$ is of order   $\sqrt{k}$, which is the same as that of the full signal space $\calK_{1,k}$ up to a logarithmic factor.} We rely on Lemma \ref{lem:construct} to accomplish this under the condition in (\ref{goodcondition}), which reads
    \begin{align}\label{l1dcon2}
        \frac{s}{m\sqrt{\log m}}\gtrsim \frac{r}{\lambda}
    \end{align}
    under $\rho\asymp \lambda^{-1}$ and $L\asymp 1$. Combining (\ref{l1dcon1}) and (\ref{l1dcon2}), we conclude that the maximal $r$ is at the order of $(\frac{\lambda k}{m\sqrt{\log m}})^{1/3}$, which leads to our final lower bound. 
    \item (Theorem \ref{thm:1bcslower} \& Corollary \ref{cor:1bcssharp}) In the original 1bCS problem the signal space is $\calK_{1,k}^*$, which does not immediately contain a small Euclidean ball. Our remedy is to introduce a lifting map $\Phi_{s}:\mathbb{B}_2^s\to \mathbb{S}^{n-1}$ in (\ref{liftingmap}), then Lemma \ref{lem:lifting} ensures the embedding of $r\mathbb{B}_2^s$ into $\calK^*_{1,k}$ through   the lifting map $\Phi_s$, i.e., $\Phi_{s}(r \mathbb{B}_2^s) \subset \calK_{1,k}^*$, under the condition (\ref{con1}) below.   Under a suitably large $k$, (\ref{con1}) reads
    \begin{align}\label{1bcscon1a}
        s\lesssim \frac{k}{r^2}. 
    \end{align}
    Since $\Phi_s$ preserves Euclidean distance as in (\ref{lowerPhi}), all that remains is to find  $\bu^* \in r\mathbb{S}^{s-1}$ such that $\sign(\bA\Phi_s(\bu^*))=\sign(\bA\Phi_s(\bm{0}))$. It turns out that this can be accomplished via Lemma \ref{lem:construct} as long as (\ref{goodcondition}) holds for $\rho,L\asymp 1$, i.e.,
    \begin{align}\label{1bcscon1b}
        \frac{s}{m\sqrt{\log m}} \gtrsim r.
    \end{align}
    Combining (\ref{1bcscon1a}) and (\ref{1bcscon1b}) establishes the desired lower bound.
\end{enumerate}

\subsection{Proofs of Theorem \ref{thm:d1bcslower}, Corollary \ref{cor:d1bcssharp}}
\begin{proof}[Proof of Theorem \ref{thm:d1bcslower}]
For the given $r\in [\sqrt{\frac{k}{m\wedge n}},1]$ and $k\ge 1$, we set 
\begin{align}\label{schoice}
    s = \Big\lfloor \frac{k}{r^2}\Big\rfloor \in [m\wedge n].
\end{align}
%By identifying a vector in $\mathbb{R}^s$ with the first $s$ entries of a vector in $\mathbb{R}^n$, we have
%\begin{align}\label{embedd1b}
%    r\mathbb{B}_2^s \subset \calK_{1,k}
%\end{align}
%due to $r\le 1$ and 
%\begin{align}
%    \forall u\in r\mathbb{B}_2^s,\quad \|\bu\|_1 \le r\sqrt{s} \le r\sqrt{k/r^2} = \sqrt{k}.
%\end{align}
Under Assumptions \ref{assump1}--\ref{assump2},
we invoke Lemma \ref{lem:construct} with $\bZ=\bA_{[s]}$, $\bxi =  \btau$ and (\ref{schoice}), yielding the following statement: if 
\begin{align}\label{ifthm1proof}
    \frac{\lfloor r^{-2}k\rfloor}{m\sqrt{\log m}}\ge 16\rho_\tau rL,
\end{align}
then with probability at least $1-\frac{2}{m}-\exp(-\frac{1}{12}\lfloor \frac{k}{r^2}\rfloor)$, there exists $\bu^*\in \mathbb{S}^{s-1}$ such that
\begin{align}\label{ustarsat}
    \sign(\bA_{[s]}(r\bu^*)+\btau) = \sign(\btau).
\end{align}
Under $r\le 1$ and $k\ge 2$, we have $\lfloor \frac{k}{r^2}\rfloor \ge \frac{k}{2r^2}$ and hence (\ref{thm1con}) guarantees (\ref{ifthm1proof}). The above conclusion then ensures the existence of $\bu^*\in \mathbb{S}^{s-1}$ such that (\ref{ustarsat}) holds with the promised probability.

All that remains is to show that such $\bu^*$ leads to the lower bound on uniform recovery error. Let $\tilde{\bu}^*=((\bu^*)^\top ,\bm{0}^{n-s})^\top \in \mathbb{S}^{n-1}$, then (\ref{ustarsat}) yields
\begin{align}\label{d1bcsequalmea}
    \sign(\bA (r\tilde{\bu}^*)+\btau )=\sign(\btau)=\sign(\bA\bm{0}+\btau).
\end{align}
Note that $r\tilde{\bu}^* \in \mathbb{B}_2^n$ since $r\le 1$. Combining with
\begin{align*}
    \|r\tilde{\bu}^*\|_1 = \|r\bu^*\|_1 \le r\sqrt{s} \stackrel{(\ref{schoice})}{\le} r\sqrt{\frac{k}{r^2}} \le \sqrt{k},
\end{align*}
we conclude that $r\tilde{\bu}^*\in \calK_{1,k}$. It is evident that $\bm{0}\in\calK_{1,k}$. In light of (\ref{d1bcsequalmea}), any estimator $\hat{\bx}$ that is a measurable function of $(\bA,\btau,\by)$ returns the same estimate for $\bm{0}$ and $r\tilde{\bu}^*$. Denoting this estimate by $\hat{\bx}$, triangle inequality gives  
$
    r = \|r\tilde{\bu}^*\|_2 \le \|\hat{\bx}-r\tilde{\bu}^*\|_2 + \|\hat{\bx}\|_2 $, which yields
\begin{align}\label{yieldlower}
   \|\hat{\bx}-\bm{0}\|_2\vee \|\hat{\bx}-r\tilde{\bu}^*\|_2\ge \frac{r}{2}.
\end{align}
This leads to the claim (\ref{thm1d1bcsconclu}) since $\bm{0},\, r\tilde{\bu}^*\in \calK_{1,k}$.
\end{proof}

We now prove Corollary \ref{cor:d1bcssharp}. 

\begin{proof}[Proof of Corollary \ref{cor:d1bcssharp}] 
For $\btau\sim {\rm Unif}[-\lambda,\lambda]^m$, Assumption \ref{assump2} holds with $\rho_{\tau} = \frac{1}{2\lambda}$. Theorem \ref{thm:d1bcslower} hence asserts that, if 
\begin{align} \label{thm1cons}
 r\in \bigg[\sqrt{\frac{k}{m\wedge n}},1\bigg]\quad\textrm{and}\quad r^3\le \frac{\lambda k}{16Lm\sqrt{\log m}},  
\end{align}
then with the promised probability $\sup_{\bx\in \calK_{1,k}}\|\hat{\bx}-\bx\|_2 \ge \frac{r}{2}$. Setting 
\begin{align*}
    r = \bigg(\frac{\lambda k}{16Lm\sqrt{\log m}}\bigg)^{1/3}
\end{align*}
ensures the second condition of (\ref{thm1cons}), and it is easy to verify that the conditions in (\ref{cond1bcs}) ensure the first condition in (\ref{thm1cons}). The proof is now complete. 
\end{proof}

\subsection{Proofs of Theorem \ref{thm:1bcslower}, Corollary \ref{cor:1bcssharp}}
Given some $r\in(0,1]$, $s\in [n-1]$, and $\bu\in \mathbb{B}_2^s$, we define the lifting map $\Phi_{s}:\mathbb{B}_2^s \to \mathbb{S}^{n-1}$ as
\begin{align}\label{liftingmap}
    \Phi_{s}(\bu) = \frac{
    (\bu^\top ,\bm{0}^{n-s-1},1)^\top 
    }{\sqrt{\|\bu\|_2^2+1}},
\end{align}
which lifts $\bu$ from dimension $s$ to dimension $n$ by appending $(\bm{0}^{n-s-1},1) = (0,\cdots,0,1)$ with $n-s-1$ zeros, followed by a normalization in $\ell_2$ norm. We notice that
\begin{align*}
    \Phi_{s}(\bm{0}) = (\bm{0}^{n-1},1)^\top  = \be_n\in \calK_{1,k}^*
\end{align*}
as long as $k\ge 1$. The following lemma states that $\Phi_{s}$ preserves the Euclidean distance between $\bu$ and $\bm{0}$ up to a factor of $\frac{1}{\sqrt{2}}$, and that 
$\Phi_{s}(r\mathbb{B}_2^s)\subset \calK^*_{1,k}$ when $s$ is not too large.

\begin{lem}[Lifting lemma] \label{lem:lifting}
Let $r\in(0,1]$, $s\in [n-1]$, $k\ge 4$, and $\bu\in \mathbb{B}_2^s$, then $\Phi_{s}$ defined in (\ref{liftingmap}) satisfies
\begin{align}\label{lowerPhi}
    \|\Phi_{s}(\bu)-\Phi_{s}(\bm{0})\|_2 \ge \frac{\|\bu\|_2}{\sqrt{2}}.
\end{align}
If moreover 
\begin{align}\label{con1}
    s\le \frac{(\sqrt{k}-1)^2}{r^2},
\end{align}
then 
\begin{align*}
    \Phi_{s}(r \mathbb{B}_2^s) \subset \calK_{1,k}^*.
\end{align*}
\end{lem}
\begin{proof} 
    Some algebra finds that 
    \begin{align*}
        \|\Phi_{s}(\bu)-\Phi_{s}(\bm{0})\|_2^2 &= \frac{\|\bu\|_2^2}{\|\bu\|_2^2+1} + \bigg(\frac{1}{\sqrt{\|\bu\|_2^2+1}}-1\bigg)^2 \\
        & = \frac{2\|\bu\|_2^2}{\sqrt{\|\bu\|_2^2+1}\cdot (\sqrt{\|\bu\|_2^2+1}+1)}.
    \end{align*}
    By   $\bu\in \mathbb{B}_2^s$, 
    \begin{align*}
        \|\Phi_{s}(\bu)-\Phi_{s}(\bm{0})\|_2^2 \ge \frac{2\|\bu\|_2^2}{\sqrt{2}(\sqrt{2}+1)} \ge \frac{\|\bu\|_2^2}{2}.
    \end{align*}
        (\ref{lowerPhi}) follows immediately.

        In light of    $\Phi_{s}(\bu)\in \mathbb{S}^{n-1}$, to guarantee $\Phi_{s}(r\mathbb{B}_2^s)\in\calK_{1,k}^*=\sqrt{k}\mathbb{B}_1^n\cap \mathbb{S}^{n-1}$, we only need $\|\Phi_{s}(\bu)\|_1\le \sqrt{k}$ for all $\bu\in r\mathbb{B}_2^s$. Note that for any $\bu\in r\mathbb{B}_2^s$, 
        \begin{align*}
            \|\Phi_{s}(\bu)\|_1 \le \|\bu\|_1 + 1 \le r\sqrt{s}+1, 
        \end{align*}
        and hence it is enough to have $r\sqrt{s}+1 \le \sqrt{k}$, which is exactly (\ref{con1}). 
\end{proof}

With the lifting lemma, we are ready to prove Theorem \ref{thm:1bcslower} by constructing $\bu^* \in r\mathbb{S}^{s-1}$ such that $\Phi_{s}(\bu^*)$ and $\Phi_{s}(\bm{0})$ are indistinguishable:
\begin{align}\label{indis1bcs}
    \sign\big(\bA\Phi_{s}(\bu^*)\big) = \sign\big(\bA\Phi_{s}(\bm{0})\big). 
\end{align}
Due to the homogeneity of $\sign$ and the definition of $\Phi_{s}$, the problem again reduces to the construction in Lemma \ref{lem:construct}.

\begin{proof}[Proof of Theorem \ref{thm:1bcslower}] Note that 
$r\ge \frac{\sqrt{k}}{\sqrt{(m\wedge n)-1}}$ ensures 
\begin{align}
    \frac{(\sqrt{k}-1)^2}{r^2}\le (m\wedge n)-1.\label{smallscales}
\end{align}
For some integer $s$ satisfying \begin{align}
    1\le s\le \frac{(\sqrt{k}-1)^2}{r^2},\label{srange1bcs}
\end{align}
it is enough to construct $\bu^*\in r\mathbb{S}^{s-1}$ such that (\ref{indis1bcs}) holds. In fact, once this is accomplished,   Lemma \ref{lem:lifting} then gives $\|\Phi_{s}(\bu^*)-\Phi_{s}(\bm{0})\|_2\ge \frac{r}{\sqrt{2}}$ and $\Phi_{s}(\bu^*),\Phi_{s}(\bm{0})\in \calK_{1,k}^*$, and therefore an argument parallel to (\ref{yieldlower}) yields the desired (\ref{uniformlowerthm2}). By  the homogeneity of $\sign$ and the definition of $\Phi_{s}$,
\begin{align*}
    (\ref{indis1bcs})\iff\sign\big(\bA_{[s]}\bu^*+\bA_{\{n\}}\big) = \sign(\bA_{\{n\}}),
\end{align*}
where $\bA_{[s]}$ denotes the sub-matrix of $\bA$ formed by the first $s$ columns, and $\bA_{\{n\}}$ denotes the last column of $\bA$. Since $s\le (m\wedge n)-1$ holds by (\ref{smallscales}) and (\ref{srange1bcs}), we can invoke Lemma \ref{lem:construct} with $\bZ=\bA_{[s]}$ and $\bxi=\bA_{\{n\}}$. Also, notice that under Assumption \ref{assump3}, the conditions of Lemma \ref{lem:construct} are satisfied. Lemma \ref{lem:construct} therefore asserts that, if   (\ref{goodcondition}) holds, then we are able to find the desired $\bu^*\in r\mathbb{S}^{s-1}$ with probability at least $1-\frac{2}{m}-\exp(-\frac{s}{12})$. We now choose the largest $s$ in  (\ref{srange1bcs}), i.e., 
\begin{align*}
    s = \bigg\lfloor \frac{(\sqrt{k}-1)^2}{r^2}\bigg\rfloor.
\end{align*}
Under $k\ge 8$ it is easy to verify that $s\ge \frac{k}{4r^2}$, and hence (\ref{goodcondition})  is guaranteed by (\ref{1bcsscaconcon}) that we assume, and also the probability term $\exp(-\frac{s}{12})$ can be relaxed to $\exp(-\frac{k}{48})$. The proof is now complete.
\end{proof}

Corollary \ref{cor:1bcssharp} follows straightforwardly from Theorem \ref{thm:1bcslower}. 
\begin{proof}[Proof of Corollary \ref{cor:1bcssharp}] For $\bA\sim N^{m\times n}(0,1)$, Assumption \ref{assump3} holds with $L=1$ and $\rho = \frac{1}{\sqrt{2\pi}}$.  Hence by Theorem \ref{thm:1bcslower}, if 
\begin{align}\label{coro2cons}
    r\in\bigg[\sqrt{\frac{k}{(m\wedge n)-1}},1\bigg]\quad \textrm{and}\quad \frac{k}{m\sqrt{\log m}} \ge \frac{64}{\sqrt{2\pi}} r^3, 
\end{align}
then with the promised probability (\ref{uniformlowerthm2}) holds.  Setting 
\begin{align*}
    r = \frac{(2\pi)^{1/6}}{4}\bigg(\frac{k}{m\sqrt{\log m}}\bigg)^{1/3}
\end{align*}
ensures the second condition of (\ref{coro2cons}), then the conditions in (\ref{1bcscons}) guarantee the  first condition in (\ref{coro2cons}). The proof is complete. 
\end{proof}

\section{Extensions}\label{sec:extension}
To ensure readability, we do not pursue full generality of the results in previous sections. We show in this section that our results can be extended toward a few directions.  We will focus on delivering the main ideas and conclusions but leave some of the details to avid readers. Note that we write $T_1\asymp T_2$ or $T_1=\Theta(T_2)$ if $C_2T_1\le T_2\le C_1T_1$ holds for some universal constants $C_1>C_2>0$, and we use $\widetilde{\Theta}$ to hide logarithmic factors in $(m,n)$. For concreteness, we typically work with $\bA\sim N^{m\times n}(0,1)$ and $\btau\sim {\rm Unif}[-\lambda,\lambda]^m$, but for the lower bounds it should be straightforward to relax these (e.g., to Assumptions \ref{assump1}--\ref{assump2} or Assumption \ref{assump3}).

\subsection{$\ell_q$ Sparse Signals}\label{sec:lq}
Our main results establish the near-optimal lower bound $\widetilde{\Omega}((k/m)^{1/3})$ for the recovery of $\ell_1$ sparse vectors. As recalled in the `Introduction,' before our work, the only lower bound is $\Omega(k/m)$ for the recovery of exactly $k$-sparse vectors in $\Sigma^n_k$, and this lower bound is attainable by NBIHT. In light of this, there is an essential gap between the optimal rates in recovery of exactly sparse signals and $\ell_1$ approximately sparse signals.

Here we smoothly bridge these two cases by considering $\ell_q$  sparse vectors for $q\in[0,1]$, implicitly for some sparsity level $k$. For $q\in(0,1)$, we define the $\ell_q$ $k$-sparse vectors in $\mathbb{B}_2^n$ as
\begin{align} \label{setdef1}
    \calK_{q,k}:= k^{\frac{1}{q}-\frac{1}{2}}\mathbb{B}_q^n\cap \mathbb{B}_2^n, 
\end{align}
which recovers the set of $\ell_1$ sparse vectors when $q=1$. Notice that 
\begin{align*}
    \bu=(u_i)_{i=1}^n\in k^{\frac{1}{q}-\frac{1}{2}}\mathbb{B}_q^n \,\iff\, \sum_{i=1}^n|u_i|^q \le k ^{1-\frac{q}{2}}, 
\end{align*}
the constraint   $\bu\in  k^{\frac{1}{q}-\frac{1}{2}}\mathbb{B}_q^n $ reduces to $\bu\in \Sigma^n_k$  (i.e., the exact $k$-sparsity) when $q\to 0$. As such, we let
\begin{align}\label{setdef2}
    \calK_{0,k}:= \Sigma^n_k \cap \mathbb{B}_2^n.
\end{align}
In the canonical 1bCS model without dithering (\ref{1bcs}), we have to assume $\bx\in\mathbb{S}^{n-1}$. Hence we work with the set of signals  
\begin{align}\label{setdef3}
    \calK_{q,k}^* :=\calK_{q,k}\cap \mathbb{S}^{n-1},\quad \forall q\in[0,1].  
\end{align}
\begin{rem} It is easy to verify that, for any $0\le q_1\le q_2\le 1$, \[ \calK_{q_1,k}\subset \calK_{q_2,k} \qquad\text{and}\qquad \calK_{q_1,k}^*\subset \calK_{q_2,k}^* . \] Hence, larger values of $q$ correspond to weaker sparsity. Moreover, $\calK_{q,k},~q\in(0,1]$ can be viewed as the tightest $\ell_q$ relaxation of $\calK_{0,k}:=\Sigma^n_k\cap \mathbb{B}_2^n$ in the sense that the equal-magnitude $k$-sparse vector $(k^{-1/2},\cdots,k^{-1/2},0,\cdots,0)$ lives in the boundary of $\calK_{q,k}$.  
\end{rem}

We pause to review existing upper bounds for the recovery of $\ell_q$ sparse vectors. For the recovery of $\bx\in \calK_{q,k}^*$ from $\bA\sim N^{m\times n}(0,1)$ and $\by=\sign(\bA\bx)$, the best known uniform recovery error rate is
\begin{align}\label{1bcslqupper}
    \sup_{\bx\in \calK_{q,k}^*}\|\hat{\bx}-\bx\|_2 =\widetilde{O}\Big(\big(\frac{k}{m}\big)^{\frac{2-q}{2+q}}\Big).
\end{align}
 In the dithered model, the best known error rate of recovering $\bx\in \calK_{q,k}$ from $\bA\sim N^{m\times n}(0,1)$, $\btau\sim {\rm Unif}[-\lambda,\lambda]^m$ and $\by=\sign(\bA\bx+\btau)$ is 
\begin{align}\label{d1bcslqupper}
    \sup_{\bx\in\calK_{q,k}}\|\hat{\bx}-\bx\|_2 =\widetilde{O}\Big(\big(\frac{\lambda k}{m}\big)^{\frac{2-q}{2+q}}\Big).
\end{align} 
The two upper bounds are attained by hamming distance minimization over $\calK^*_{q,k}$ \cite{oymak2015near,dirksen2021non,chen2024optimal} or the projected gradient descent algorithm with projection onto $k^{\frac{1}{q}-\frac{1}{2}}\mathbb{B}_q^n$ \cite{chen2024optimal}. Specifically, \cite[Remark 7]{chen2024optimal} shows that (\ref{1bcslqupper}) is achieved by Hamming distance minimization over $\calK_{q,k}^*$ whenever $m=\widetilde{\Omega}(k)$, and that (\ref{d1bcslqupper}) is achieved by Hamming distance minimization over $\calK_{q,k}$  whenever $m=\widetilde{\Omega}(\lambda^{-\frac{2q}{2-q}}k)$.

The main goal of this subsection is to demonstrate that our techniques are capable of establishing matching lower bounds for (\ref{1bcslqupper}) and (\ref{d1bcslqupper}).  In fact, the only difference of working with $\ell_q$ sparse vectors lies in the embedding of $r\mathbb{B}_2^s$ into the space of $\ell_q$ sparse vectors, where the corresponding  conditions (\ref{l1dcon1}) and (\ref{1bcscon1a}) should be generalized. This is made precise by the following lemma, whose proof is relegated to Appendix \ref{app:provelem} to preserve the presentation flow.  

\begin{lem}[Embedding of $r\mathbb{B}_2^s$ into $\calK_{q,k}$ and $\calK_{q,k}^*$] \label{lem:lqadjust}Let $r\in (0,1]$, $q\in [0,1]$, and $s\in [n]$. Then we have the following for $\calK_{q,k}$ and $\calK_{q,k}^*$ defined in (\ref{setdef1}), (\ref{setdef2}) and (\ref{setdef3}):
\begin{itemize}
    \item If 
    \begin{align}\label{d1bcslqcon}
        s\le \frac{k}{r^{\frac{2q}{2-q}}},
    \end{align}
    then by identifying $\mathbb{R}^s$ with the first $s$ coordinates of $\mathbb{R}^n$, 
    \begin{align*}
        r\mathbb{B}_2^s\subset \calK_{q,k}.
    \end{align*}
    \item If $s\le n-1$ and 
    \begin{align}\label{1bcslqcon}
        s \le \bigg(\frac{k^{\frac{2-q}{2}}-1}{r^q}\bigg)^{\frac{2}{2-q}}, 
    \end{align}
    then 
    \begin{align*}
        \Phi_s(r\mathbb{B}_2^s)\subset\calK_{q,k}^*. 
    \end{align*}
\end{itemize}
\end{lem}

The lower bounds then readily follow. For the recovery of $\bx\in \calK_{q,k}^*$ from $\by=\sign(\bA\bx+\btau)$ where $\bA\sim N^{m\times n}(0,1)$ and $\btau\sim {\rm Unif}[-\lambda,\lambda]^m$, we combine the ``embedding condition'' (\ref{d1bcslqcon}) with the ``construction condition'' (\ref{l1dcon2}) to yield that the maximal $r$ is at the order of 
\begin{align*}
     \bigg(\frac{\lambda k}{m\sqrt{\log m}}\bigg)^{\frac{2-q}{2+q}},
\end{align*}
which also dictates the order of our lower bound. Hence, our lower bound is near-optimal in that it matches the upper bound in (\ref{d1bcslqupper}), up to logarithmic factors. By generalizing Theorem \ref{thm:d1bcslower} and Corollary \ref{cor:d1bcssharp} to $\ell_q$ sparse vectors, our main conclusion is the following:

\medskip

\noindent\textbf{Takeaway.} {\it For the recovery of $\bx\in \calK_{q,k}$ from $\by=\sign(\bA\bx+\btau)$ with   $\bA\sim N^{m\times n}(0,1)$, $\btau\sim {\rm Unif}[-\lambda,\lambda]^m$ for some $q\in[0,1]$, the optimal uniform recover error rate is (in general) at the order of $\widetilde{\Theta}((\frac{\lambda k}{m})^{\frac{2-q}{2+q}})$.}

We similarly treat the recovery of $\bx\in \calK_{q,k}^*$ from $\by=\sign(\bA\bx)$ with Gaussian matrix $\bA$. In light of the new ``embedding condition'' (\ref{1bcslqcon})---which reads $s\lesssim kr^{-\frac{2q}{2-q}}$ as long as $k$ is suitably large (e.g., $k\ge 2^{\frac{2}{2-q}}$)---and the same ``construction condition'' in (\ref{1bcscon1b}), we are able to prove a lower bound at the order of 
\begin{align*}
     \bigg(\frac{k}{m\sqrt{\log m}}\bigg)^{\frac{2-q}{2+q}}. 
\end{align*}
Over all $q\in[0,1]$, this matches the upper bound (\ref{1bcslqupper}) up to logarithmic factors and hence is near-optimal. Therefore, we reach the following takeaway message:

\medskip

\noindent\textbf{Takeaway.} {\it For the recovery of $\bx\in \calK^*_{q,k}$ from $\by=\sign(\bA\bx)$ with   $\bA\sim N^{m\times n}(0,1)$ for some $q\in[0,1]$, the optimal uniform recover error rate is, in general, at the order of $\widetilde{\Theta}\left(\big(\frac{k}{m}\big)^{\frac{2-q}{2+q}}\right)$.}

\begin{rem} \label{rem:recover}
For $q=0$, we recover the only existing lower bound $\Omega(\frac{k}{m})$ for 1bCS of exactly sparse signals, up to a logarithmic factor. We point out that in \cite{jacques2013robust,acharya2017improved}  this lower bound was established by a deterministic argument based on the pigeonhole principle, thus it holds for arbitrary matrix $\bA$. In contrast, we provide a novel  probabilistic argument, which yields near-optimal lower bounds over all $q\in [0,1]$ but instead only applies to a certain class of random matrices.  
\end{rem}

\subsection{Designs with Sub-Weibull Tails}\label{sec:subweibull}
As mentioned in Remark \ref{rem:recover}, the previously known lower bound $\Omega(k/m)$ is valid for arbitrary designs, while our new lower bounds rely on independent, sub-Gaussian $a_i$'s (see Assumptions \ref{assump1}--\ref{assump3}). To narrow this gap, we seek to relax the assumptions of the designs in our lower bounds.

We point out that
\textit{our lower bounds remain valid for independent $\ba_i$'s with general sub-Weibull tails.} See  \cite{kuchibhotla2022moving} for a brief introduction of sub-Weibull random variables. More precisely, (\ref{subgauss1}) in Assumption \ref{assump1} for the dithered model can be relaxed to 
\begin{align}\label{subweibull}
    \mathbb{P}\big(|\ba_i^\top \bu|\ge t\big)\le 2\exp\bigg(-\frac{1}{2}\Big(\frac{t}{L}\Big)^\alpha\bigg),\quad \forall t>0
\end{align}
for general $\alpha>0$. To see this, the only step in our proofs that relies on   {\it sub-Gaussianity}  is (\ref{sgmaxupper}). Under (\ref{subweibull}), for some $\bu\in \mathbb{S}^{n-1}$ oblivious to $\bA$ we still have 
\begin{align*}
    \mathbb{P}\bigg(\max_{i\in [m]}\,|\ba_i^\top \bu|\le L\big(4\log m\big)^{1/\alpha}\bigg) \ge1-\frac{2}{m}. 
\end{align*}
As such, for the recovery of $\bx\in \calK_{1,k}$ from $\by=\sign(\bA\bx+\btau)$,  our argument yields a lower bound at the order of 
\begin{align*}
    \bigg(\frac{\lambda k}{Lm\log ^{1/\alpha}m}\bigg)^{1/3}
\end{align*}
under $\bA$ obeying (\ref{subweibull}) and $\btau\sim {\rm Unif}[-\lambda,\lambda]^m$.

In parallel, for the recovery of $\bx\in \calK^*_{1,k}$ from $\by=\sign(\bA\bx)$ where the rows of  $\bA$ are independent and have sub-Weibull tails (and the last column of $\bA$ satisfies some conditions similar to Assumption \ref{assump3}), our techniques lead to a $\Omega((\frac{k}{Lm\log^{1/\alpha}m})^{1/3})$ lower bound.

Finally, we note in passing that the lower bound $\widetilde{\Omega}((\lambda k/m)^{1/3})$ for the dithered model also holds for {\it arbitrary bounded designs} satisfying $\|\ba_i\|_2=\widetilde{O}(1),~i\in[m]$, which is widely studied in machine learning (e.g., \cite{kakade2011efficient}). To see this, after the construction of $\bu^*\in \mathbb{S}^{n-1}$ in (\ref{constructustar}), the worst-case bound for bounded designs 
\begin{align*}
    \max_{i\in \calI^c} |\bz_i^\top \bu|\le  \max_{i\in \calI^c}\|\bz_i\|_2 \le  \max_{i\in \calI^c} \|\ba_i\|_2 =\widetilde{O}(1)
\end{align*}
can play the role of (\ref{sgmaxupper}). However, it is unclear whether our lower bound extends to deterministic designs with   rows of $\Theta(\sqrt{n})$ Euclidean norm.

\subsection{Adversarial Bit Flips}\label{sec:adbitflip}
We now move on to robust one-bit compressed sensing under \emph{adversarial bit flips}, a widely adopted corruption model in the area \cite{plan2012robust,matsumoto2024robust,dirksen2023robust,awasthi2016learning,dirksen2021non,chinot2022adaboost}. Instead of observing the noiseless binary vector $\by$, one observes $\hat{\by}\in\{-1,1\}^m$ obeying 
\begin{align}\label{betadiffer}
    d_{\rm H}\big(\hat{\by},\by\big) \le \beta m,
\end{align}
where $\beta\le c_*$ for some small enough universal constant $c_*>0$. Note that such $\hat{\by}$ can be chosen by an adversary having knowledge of $(\bA,\btau,\bx)$, as long as it satisfies (\ref{betadiffer}).

In the dithered model $\by=\sign(\bA\bx+\btau)$, the best known upper bound on uniform recovery error over $\calK_{1,k}$  is 
\begin{align}\label{rd1bcsupper}
    \sup_{\bx\in \calK_{1,k}}\|\hat{\bx}-\bx\|_2=\widetilde{O}\bigg(\big(\frac{\lambda k}{m}\big)^{1/3}+\lambda \beta\bigg)
\end{align} and can be achieved by, e.g., a convex program under $m=\widetilde{\Omega}(\lambda^{-2}k)$  \cite[Theorem 1]{jung2021quantized}. In the canonical 1bCS model $\by=\sign(\bA\bx)$ with $\beta$-fraction adversarial bit flips, the best known upper bound for uniform recovery error over $\calK_{1,k}^*$ is 
\begin{align}\label{r1bcsupper}
    \sup_{\bx\in \calK_{1,k}^*}\|\hat{\bx}-\bx\|_2=\widetilde{O}\bigg(\big(\frac{k}{m}\big)^{1/3}+\beta\bigg)
\end{align}
and is achieved by, e.g., projected gradient descent under $m=\widetilde{\Omega}(k)$ \cite[Section 5.1]{chen2024optimal}. As such, the effect of the adversarial bit flips is captured by the term $\widetilde{O}(\lambda \beta)$ in (\ref{rd1bcsupper}), and the term $\widetilde{O}(\beta)$ in (\ref{r1bcsupper}).

We now provide a simple argument to show the tightness of the error increment terms. For concreteness, we consider the dithered model $\by=\sign(\bA\bx+\btau)$ with $\bA\sim N^{m\times n}(0,1)$ and $\btau\sim {\rm Unif}[-\lambda,\lambda]$. Choosing $\bx_1,\bx_2\in \calK_{1,k}$ satisfying (suppose $\lambda\beta\le 2$ for simplicity)
\begin{align*}
    \|\bx_1-\bx_2\|_2 =  \lambda\beta, 
\end{align*} 
we recall the known bound (e.g., \cite[Section 4.2]{chen2024optimal})
\begin{align*}
    \mathbb{P}\bigg(\sign(\ba_i^\top \bx_1+\btau_i)\ne\sign(\ba_i^\top \bx_2+\btau_i)\bigg)\le \frac{\|\bx_1-\bx_2\|_2}{2\lambda} = \frac{\beta}{2}.
\end{align*}
Hence Chernoff bound gives that
\begin{align*}
    \mathbb{P}\bigg(d_{\rm H}\big(\sign(\bA\bx_1+\btau),\sign(\bA\bx_2+\btau)\big) \le \frac{3\beta m}{4}\bigg) \ge 1-\exp(-c\beta m).
\end{align*}
On this event, one may observe the same vector $\hat{\by}$ for $\bx_1,\,\bx_2$---for instance, if the adversary changes no observations when the true signal is $\bx_1$ but corrupts the observations to $\sign(\bA\bx_1+\btau)$ when the true signal is $\bx_2$. If this happens, then any estimator $\hat{\bx}=\hat{\bx}(\bA,\btau,\hat{\by})$ returns the same estimate for $\bx_1,\,\bx_2$, thus incurring uniform recovery error not smaller than $\frac{\lambda\beta}{2}$. A similar argument establishes a lower bound $\Omega(\beta)$ for the uniform recovery error in the recovery of $\bx\in \calK_{1,k}^*$ from $\bA\sim N^{m\times n}(0,1)$ and $\hat{\by}$ obeying $d_{\rm H}(\hat{\by},\sign(\bA\bx))\le \beta m$. As our main results in Section \ref{sec:mainresults} imply the sharpness of $\widetilde{O}((\frac{\lambda k}{m})^{1/3})$ in (\ref{rd1bcsupper}) and $\widetilde{O}((\frac{k}{m})^{1/3})$ in (\ref{r1bcsupper}), we reach the following conclusion: 

\medskip

\noindent\textbf{Takeaway.} {\it Consider the corruption model of $\beta$-fraction adversarial bit flips.
In 1bCS of $\bx\in \calK_{1,k}^*$ with $\bA\sim N^{m\times n}(0,1)$, the optimal uniform recovery error rate is $\widetilde{\Theta}\left(\big(\frac{k}{m}\big)^{1/3}+\beta\right)$; in D1bCS of $\bx\in\calK_{1,k}$ with $\bA\sim N^{m\times n}(0,1)$ and $\btau\sim {\rm Unif}[-\lambda,\lambda]^m$, the optimal uniform recovery error rate is $\widetilde{\Theta}\left(\big(\frac{\lambda k}{m}\big)^{1/3}+\lambda\beta\right)$.}

\subsection{Low-Rank Matrices}\label{sec:matrix}
Beyond sparsity, our results can be extended to (approximately) low-rank matrices. We denote the set of $n_1\times n_2$ rank-$\bar{r}$ matrices     by $M^{n_1,n_2}_{\bar{r}}:=\{\bU\in\mathbb{R}^{n_1\times n_2}:\rank(\bU)
\le \bar{r}\}$, the Frobenius norm ball in $\mathbb{R}^{n_1\times n_2}$ by $\mathbb{B}_{\rm F}^{n_1,n_2}:=\{\bU\in\mathbb{R}^{n_1\times n_2}:\|\bU\|_{\rm F}\le 1\}$, the Frobenius norm sphere in a matrix space by $\mathbb{S}_{\rm F}$ (with dimension depending on the context), and the nuclear norm ball in $\mathbb{R}^{n_1\times n_2}$ by $\mathbb{B}_{\rm nu}^{n_1,n_2}:=\{\bU\in \mathbb{R}^{n_1\times n_2}:\|\bU\|_{\rm nu}\le 1\}$.\footnote{The nuclear norm of a matrix $\bU$, denoted by $\|\bU\|_{\rm nu}$, is defined as the sum of all the singular values of $\bU.$} Without loss of generality, we assume $n_1\ge n_2$.

The recovery of $\bX\in M_{\bar{r}}^{n_1,n_2}$ has already been resolved. In the dithered model concerning the recovery of $\bX\in \calX_0:=M^{n_1,n_2}_{\bar{r}}\cap \mathbb{B}_{\rm F}^{n_1,n_2}$ from $\{\bA_i\}_{i=1}^m\stackrel{iid}{\sim}N^{n_1\times n_2}(0,1)$, $\btau\sim {\rm Unif}[-\lambda,\lambda]^m$ and $\{y_i=\sign(\langle \bA_i, \bX\rangle+\tau_i)\}_{i=1}^m$,\footnote{In this subsection, we suppose $\lambda=\Theta(1)$ is large enough for the dithered model and hence we do not track the dependence on $\lambda$.} or the canonical model concerning recovery of $\bX\in\calX_0^*:= M^{n_1,n_2}_{\bar{r}}\cap \mathbb{S}_{\rm F}$ from $\{\bA_i\}_{i=1}^m\stackrel{iid}{\sim}N^{n_1\times n_2}(0,1)$ and $\{y_i=\sign(\langle \bA_i,\bX \rangle)\}_{i=1}^m$, under $m=\widetilde{\Omega}(\bar{r}n_1)$ the sharp uniform error rate is 
\begin{align*}
\inf_{\substack{\hat{\bX}=\hat{\bX}(\{\bA_i,\tau_i,y_i\}_{i=1}^m)\\{\rm or}\\\hat{\bX}=\hat{\bX}(\{\bA_i,y_i\}_{i=1}^m)}}\sup_{\substack{\bX\in \calX_0\\{\rm or}\\\bX\in \calX_0^*}}\, \|\hat{\bx}-\bx\|_{\rm F} =\widetilde{\Theta}\Big(\frac{\bar{r}n_1}{m}\Big).
\end{align*}
Here the upper bound is achieved by a projected gradient descent algorithm (with a projection onto $M_{\bar{r}}^{n_1,n_2}$),  and the lower bound is established by a direct adaptation of the $\Omega(k/m)$ lower bound from \cite{jacques2013robust,acharya2017improved}. See \cite[Section 2]{chen2024optimal}.

To be more practical, one can relax the exact low-rankness to approximately low-rankness captured by low nuclear norm. More precisely, we define the set of  $\ell_1$ approximately rank-$\bar{r}$ matrices in the Frobenius norm ball as 
\begin{align*}
    \calX_1:=\sqrt{\bar{r}}\mathbb{B}_{\rm nu}^{n_1,n_2}\cap \mathbb{B}_{\rm F}^{n_1,n_2}, 
\end{align*} 
and similarly define 
\begin{align*}
    \calX_1^*:= \sqrt{\bar{r}}\mathbb{B}_{\rm nu}^{n_1,n_2} \cap \mathbb{S}_{\rm F}.
\end{align*}
They serve as relaxations of $\calX_0$ and $\calX_0^*$, respectively. It was known \cite{chen2024optimal} that projected gradient descent (with projection onto $\sqrt{r}\mathbb{B}_{\rm nu}^{n_1,n_2}$) achieves the uniform recovery upper bound
\begin{align}\label{approlrupper}
    \sup_{\bX\in \calX_1\textrm{\,or\,}\calX_1^*}\|\hat{\bX}-\bX\|_{\rm F} =\widetilde{O}\Big(\big(\frac{\bar{r}n_1}{m}\big)^{1/3}\Big). 
\end{align}

The techniques of this paper are sufficient to establish a near-matching lower bound for (\ref{approlrupper}). Let us demonstrate this for the dithered model. As explained in Section \ref{sec:technicaloverview}, the irregular rate $(k/m)^{1/3}$ is dictated by the two conditions about embedding and construction. To work with $\calX_1:= \sqrt{\bar{r}}\mathbb{B}_{\rm nu}^{n_1,n_2}\cap \mathbb{B}_{\rm F}^{n_1,n_2}$, all we need is to modify the two conditions appropriately. For some $R\in [n_2]$, by identifying $\mathbb{R}^{n_1\times R}$ with the first $R$ columns of $\mathbb{R}^{n_1\times n_2}$, we are able to embed a radius-$r$ Frobenius norm ball in $\mathbb{R}^{n_1\times R}$ into $\calX_1$:
\begin{align}\label{lrembedding}
    r\mathbb{B}_{\rm F}^{n_1,R}\subset \calX_1,
\end{align}
as long as $r\le 1$ and 
\begin{align}
     r\sqrt{R}\le \sqrt{\bar{r}}\,\iff \, R\le \frac{\bar{r}}{r^2}. \label{lrembed}
\end{align}
(Note that the matrices in $r\mathbb{B}_{\rm F}^{n_1,R}$ have nuclear norms lower than $r\sqrt{R}$.) Applying Lemma~\ref{lem:construct} to the $n_1R$-dimensional small ball $r\mathbb{B}_{\rm F}^{n_1,R}$, we are able to construct $\bU^*\in r\mathbb{S}_{\rm F}$ such that 
\begin{align*}
    \sign(\langle \bA_i,\bU^*\rangle  + \tau_i) = \sign(\tau_i),\quad \forall i\in [m]
\end{align*}
so long as
\begin{align} \label{lrconstruct}
    \frac{n_1R}{m\sqrt{\log m}} \gtrsim r.
\end{align}
Combining  (\ref{lrembed}) and (\ref{lrconstruct}), we establish a lower bound at the order of the maximal $r$, i.e., 
\begin{align*}
    \bigg(\frac{\bar{r}n_1}{m\sqrt{\log m}}\bigg)^{1/3}.
\end{align*}

 Parallel adaptations, together with the lifting trick,  establish the same lower bound for the canonical 1bCS model. Since this lower bound matches the upper bound in (\ref{approlrupper}), we drop the condition of $n_1\ge n_2$ and conclude the following:

 \medskip

\noindent\textbf{Takeaway.} {\it In the recovery of $\bX\in \sqrt{\bar{r}}\mathbb{B}_{\rm nu}^{n_1,n_2}\cap \mathbb{S}_{\rm F}$ from $\{y_i=\sign(\langle \bA_i,\bX\rangle)\}_{i=1}^m$ with $\{\bA_i\}_{i=1}^m\stackrel{iid}{\sim} N^{n_1\times n_2}(0,1)$, or the recovery of $\bX\in \sqrt{\bar{r}}\mathbb{B}_{\rm nu}^{n_1,n_2}\cap \mathbb{B}_{\rm F}^{n_1,n_2}$ from $\{y_i=\sign(\langle \bA_i,\bX\rangle+\tau_i)\}_{i=1}^m$ with $\{\bA_i\}_{i=1}^m\stackrel{iid}{\sim} N^{n_1\times n_2}(0,1)$ and $\{\tau_i\}_{i=1}^m\stackrel{iid}{\sim}{\rm Unif}[-\lambda,\lambda]$ ($\lambda =  \Theta(1)$), the optimal uniform recovery error rate is $\widetilde{\Theta}\left(\big(\frac{\bar{r}\max\{n_1,n_2\}}{m}\big)^{1/3}\right)$.}

\subsection{Transition to Non-Sparse Case}\label{sec:transition}
Instead of an \emph{extension}, we present here an interesting \emph{refinement} of our lower bounds. First observe that our lower bound $\widetilde{\Omega}((\frac{k}{m})^{1/3})$ makes no sense in a transition to non-sparse case: if $k\to n$, then the lower bound approaches $\widetilde{\Omega}((\frac{n}{m})^{1/3})$, which is too large and contradicts the optimal error rate $\widetilde{\Theta}(\frac{n}{m})$ known for the non-sparse case \cite{long1995sample,acharya2017improved,hsu2024sample}. In fact, in view of $\sqrt{\frac{k}{m\wedge n}}\le r\le 1$ in Theorem \ref{thm:d1bcslower} and $ \sqrt{\frac{k}{(m\wedge n)-1}}\le r\le 1$ in Theorem \ref{thm:1bcslower}, our theorems do not allow for $k\to n$. The main goal here is to get rid of this constraint and establish lower bounds that exhibit a phase transition from $\ell_1$ sparse case to the non-sparse case.

We shall discuss the dithered model with $\bA\sim N^{m\times n}(0,1)$ and $\btau\sim{\rm Unif}[-\lambda,\lambda]^m$ with large enough $\lambda=\Theta(1)$, which satisfies Assumptions \ref{assump1}--\ref{assump2} with $L=1$ and $\rho_{\tau}=\Theta(1)$. As discussed in Section~\ref{sec:technicaloverview}, the first step is to embed $r\mathbb{B}_2^s$ into $\calK_{1,k}$, in which we need $s\le \frac{k}{r^2}$ to guarantee   $\ell_1$ norms bounded by $\sqrt{k}$. Yet, we also (implicitly) need $s\le n$, i.e., the embedding dimension is not higher than the ambient dimension, and need $s\le m$, in order to invoke Lemma \ref{lem:construct}. In the proof of Theorem~\ref{thm:d1bcslower} we used $s=\lfloor \frac{k}{r^2}\rfloor$, the largest $s$ under the first constraint, and then assume $r\ge \sqrt{\frac{k}{m\wedge n}}$ to ensure $\lfloor \frac{k}{r^2}\rfloor\le m\wedge n$; see (\ref{schoice}). As such, to remove the constraint $r\ge \sqrt{\frac{k}{m\wedge n}}$, we can instead set
\begin{align*}
    s = m\wedge n\wedge \Big\lfloor \frac{k}{r^2}\Big\rfloor.
\end{align*}
Combining with the construction condition $\frac{s}{m\sqrt{\log m}}\gtrsim r$ in (\ref{l1dcon2}),
we obtain the lower bound 
\begin{align}\label{completeorder}
    \Omega\bigg( \Big(\frac{k}{m\sqrt{\log m}}\Big)^{1/3}\wedge\frac{1}{\sqrt{\log m}}\wedge \frac{n}{m\sqrt{\log m}}\bigg) = \widetilde{\Omega}\bigg(\Big(\frac{k}{m}\Big)^{1/3}\wedge \frac{n}{m}\bigg)
\end{align}
that holds under $m\gtrsim k$. Intriguingly, this lower bound exhibits a   transition from the $\ell_1$-sparse rate $(k/m)^{1/3}$ to the non-sparse rate $\frac{n}{m}$. This transition is roughly located at $k\asymp \frac{(m\wedge n)^3}{m^2}$. Under $m=\widetilde{\Omega}(k)$, existing results   also imply an upper bound at the same order (up to logarithmic factors)---$\widetilde{O}\Big(\Big(\frac{k}{m}\Big)^{1/3}\wedge \frac{n}{m}\Big)$
achieved by (\ref{hdm1bcs}); see \cite[Remark 6]{chen2024optimal} for instance. Hence, whenever $m=\widetilde{\Omega}(k)$, a complete characterization of the optimal uniform recovery error for D1bCS over $\calK_{1,k}$ is given by (\ref{completeorder}).

Moreover, similar adaptations apply to Theorem~\ref{thm:1bcslower} and yield a fine-grained lower bound identical to (\ref{completeorder}) for canonical 1bCS. Existing works imply an upper bound of the same order (up to logarithmic factors) being achieved by (\ref{hdmd1bcs}) under $m=\widetilde{\Omega}(k)$; see \cite[Remark 6]{chen2024optimal} for instance. Therefore, the takeaway message of the discussion here is the following:

\medskip

\noindent\textbf{Takeaway.} {\it In the recovery of $\bx\in \calK_{1,k}$ from $\by=\sign(\bA\bx+\btau)$ with $\bA\sim N^{m\times n}(0,1)$ and $\btau\sim {\rm Unif}[-\lambda,\lambda]^m$ where $\lambda=\Theta(1)$, and the recovery of $\bx\in \calK_{1,k}^*$ from $\by=\sign(\bA\bx)$ with  $\bA\sim N^{m\times n}(0,1)$, the optimal error rate under $m=\widetilde{\Omega}(k)$ is $\widetilde{\Theta}\left(\min\Big\{\big(\frac{k}{m}\big)^{1/3},\frac{n}{m}\Big\}\right)$.}

\subsection{Multiple Extensions}
While we have discussed a number of extensions separately, we shall point out that  some of these extensions can be pursued simultaneously. For example, by combining the ideas in Sections \ref{sec:lq} and \ref{sec:adbitflip}, one can show that in the recovery of $\ell_q$ ($q\in[0,1]$) sparse vectors in canonical and uniformly dithered  models, the optimal error rate under $\beta m$ adversarial bit flips is at the order of
\begin{align*}
    \widetilde{\Theta}\bigg(\Big(\frac{k}{m}\Big)^{\frac{2-q}{2+q}}+\beta\bigg).
\end{align*}
Similarly, one can combine the developments in Sections \ref{sec:lq} and \ref{sec:transition} to establish the  
fine-grained lower bound 
\begin{align*}
    \widetilde{\Omega}\bigg(\min\bigg\{\Big(\frac{k}{m}\Big)^{\frac{2-q}{2+q}},\frac{n}{m}\bigg\}\bigg) 
\end{align*}
for the recovery of $\ell_q$ sparse vectors, thus matching the upper bound in \cite[Remark 7]{chen2024optimal} up to logarithmic factors. %This lower bound remains valid for sub-Weibull designs, due to the observation in Section \ref{sec:subweibull}. 

\section{Conclusion}\label{sec:conclusion}
In this paper, we established near-optimal lower bounds for the uniform recovery of approximately sparse signals in the canonical 1bCS model and the model with uniform dithering. Our main results show the lower bound $\widetilde{\Omega}((\frac{k}{m})^{1/3})$ for the recovery of $\ell_1$ sparse vectors, which is sharp up to logarithmic factors due to known matching upper bounds. Our argument consists of two steps: the first step is to embed an $\ell_2$ ball in suitable dimension and of small radius into the signal set, which is straightforward for the dithered model but relies on a lifting map for the canonical model; the second step is to construct two points in the small ball that are  indistinguishable, which we accomplished by a general construction lemma. We discussed that our results extend in a few ways, most notably, they can be generalized to $\ell_q$ sparse signals.

There remain some questions that require further investigation. First of all, do our lower bounds remain valid for \emph{arbitrary designs}? Since the lower bound $\Omega(k/m)$ for exactly $k$-sparse signals holds for arbitrary designs, it is tempting to believe that the lower bound $\widetilde{\Omega}((k/m)^{1/3})$ for $\ell_1$ sparse signals also holds much more generally. Yet our current argument cannot yield this general lower bound, since
the construction 
in Lemma \ref{lem:construct} relies on the randomness of $\bA$. To attack this problem one may need to develop fundamentally different arguments. Second, even under Gaussian $\bA$ (and uniform dither), the refinement of logarithmic factor is interesting. In particular, we conjecture that the factor $\sqrt{\log m}$ in our lower bound $\Omega((\frac{k}{m\sqrt{\log m}})^{1/3})$ is simply a proof artifact. Is there a more refined argument that can remove this $\sqrt{\log m}$ factor? Moreover, there are more intricate quantized recovery models, such as dithered multi-bit compressed sensing \cite{jung2021quantized,chen2024optimal} and one-bit phase retrieval \cite{chen2024one}, where the best known upper bounds for  recovering $\ell_1$ sparse signals exhibit similar $m^{-1/3}$ decay rates. Can we establish matching lower bounds for these models? For this problem,  it appears to be promising to build upon the techniques in this paper with suitable adaptations.               

%particularly extending to $\ell_q$ sparsity, sub-Weibull designs, adversarial bit flipping, matrix recovery,

\subsection*{Acknowledgments}
Part of the work was done while JC was visiting Hal{\i}c{\i}o\u{g}lu Data Science Institute at UC San Diego in the summer of 2026. JC thanks Rayan Saab for some stimulating discussions on $\ell_q$ sparse recovery and a potential extension to $q\in(1,2)$. 

  \bibliography{libr}
\bibliographystyle{plain}

\appendix

\section{Proof of Lemma \ref{lem:lqadjust}} \label{app:provelem}
\begin{proof}[Proof of Lemma \ref{lem:lqadjust}]
 For any $q\in(0,1]$ we shall repeatedly use the elementary bound
\begin{align}
    \|\bu\|_q^q
    \le s^{1-\frac q2}\|\bu\|_2^q ,\quad \forall u\in \mathbb{R}^s\label{holders}
\end{align}
which follows from H\"older's inequality.

{\it The case of $q\in(0,1]$.} For the first claim, take any $\bu\in r\mathbb B_2^s$, and since $r\le 1$, by identifying $\bu$ with $(\bu^\top ,\bm{0}^{n-s})^\top $ then $\bu\in \mathbb{B}_2^n$. Moreover, using (\ref{holders}) and $\|\bu\|_2\le r$, 
\begin{align*}
    \|\bu\|_q^q
    \le
    s^{1-\frac q2}\|\bu\|_2^q
    \le
    s^{1-\frac q2} r^q. 
\end{align*}
Now notice that the condition (\ref{d1bcslqcon}) is equivalent to
$
    s^{1-\frac q2} r^q \le k^{1-\frac q2}.
$
Hence $\|\bu\|_q\le k^{\frac{1}{q}-\frac{1}{2}}$, meaning that $\bu\in k^{\frac{1}{q}-\frac{1}{2}}\mathbb{B}_q^n$ by viewing $\bu$ as $(\bu^\top ,\bm{0}^{n-s})^\top $. 
This proves
$r\mathbb B_2^s\subset \calK_{q,k}$.
For the second claim, take any $\bu\in r\mathbb B_2^s$ and write
$\bz=\Phi_s(\bu).$
  By (\ref{holders}), 
\begin{align*}
     \|\bz\|_q^q
    =
    \frac{\sum_{j=1}^s |u_j|^q+1}
    {(\|\bu\|_2^2+1)^{q/2}}
    \le
    \sum_{j=1}^s |u_j|^q+1
    \le
    s^{1-\frac q2}r^q+1. 
\end{align*}
Note that the condition 
is equivalent to 
$
    s^{1-\frac q2} r^q +1
    \le
    k^{1-\frac q2}.$ Therefore we reach
$\|\bz\|_q^q\le k^{1-\frac q2},$
namely 
$
    \bz\in k^{\frac1q-\frac12}\mathbb B_q^n.
$ Since also $\bz\in\mathbb S^{n-1}$ by construction, we have $\bz\in\calK_{q,k}^*$.

{\it The case of $q=0$.} Recall that
$\calK_{0,k}=\Sigma_k^n\cap\mathbb B_2^n.$
For the first claim, any $\bu\in r\mathbb B_2^s$, embedded as
$\tilde \bu=(\bu^\top ,\bm{0}^{n-s})^\top $, has support size at most $s$ and Euclidean norm
at most $r\le 1$. Since the condition (\ref{d1bcslqcon}) reduces to $s\le k$, we have $\tilde{\bu}\in\calK_{0,k}$, thus proving $r\mathbb{B}_2^n\subset \calK_{0,k}$. For the second claim, for any $\bu\in r\mathbb B_2^s$, the lifted vector
$\Phi_s(\bu)$ has support size at most $s+1$, because it may have nonzero
entries only among the first $s$ coordinates and the last coordinate.
It also has Euclidean norm one by construction. Since the condition (\ref{1bcslqcon}) now reduces to $s\le k-1$, we have 
$
    \Phi_s(\bu)\in\Sigma_k^n\cap\mathbb S^{n-1}
    =
    \calK_{0,k}^* .$ 
\end{proof}
\end{document}